\documentclass[10pt, twocolumn]{article}

\usepackage{article_macros}
\usepackage[top = 1in, bottom = 1in, left = 0.7in, right = 0.7in]{geometry}
\usepackage{algorithm}
\usepackage{algpseudocode}
\usepackage{makecell}
\usepackage[frozencache,cachedir=_minted-main]{minted}
\usepackage{tabularx}
\usepackage{booktabs}
\usepackage[tt = false]{libertinus}

\setminted{
  frame = single,     
  framesep = 2mm,     
  fontsize = \footnotesize,
}

\newcommand{\cas}{\mathrm{cas}}
\newcommand{\cht}{\mathsf{H}}
\newcommand{\qht}{\mathsf{QHT}}
\newcommand{\cft}{\mathsf{F}}
\newcommand{\qft}{\mathsf{QFT}}
\newcommand{\qct}{\mathsf{QCT}}
\newcommand{\qst}{\mathsf{QST}}
\newcommand{\cnot}{\textsc{cnot}}
\newcommand{\orop}{\textsc{or}}
\newcommand{\inc}{\textsc{inc}}
\newcommand{\dec}{\textsc{dec}}

\algrenewcommand{\Require}{\item[\textbf{Input:}]}
\algrenewcommand{\Ensure}{\item[\textbf{Output:}]}
\algrenewcommand{\alglinenumber}[1]{\footnotesize #1.}

\newcommand{\iddots}{\reflectbox{$\ddots$}}

\allowdisplaybreaks
\title{QRTlib: A Library for Fast Quantum Real Transforms}

\author{
    Armin Ahmadkhaniha \\ {\tt ahmadkha@mcmaster.ca} \and
    Lu Chen \\ {\tt chenl143@mcmaster.ca} \and
    Jake Doliskani \\ {\tt jake.doliskani@mcmaster.ca} \and
    Zhifu Sun \\ {\tt sun143@mcmaster.ca} \\[4mm]
    Department of Computing and Software, McMaster University
}

\date{}

\begin{document}

\maketitle

\paragraph{Abstract}
Real-valued transforms such as the discrete cosine, sine, and Hartley transforms play a central role in classical computing, complementing the Fourier transform in applications from signal and image processing to data compression. However, their quantum counterparts have not evolved in parallel, and no unified framework exists for implementing them efficiently on quantum hardware. This article addresses this gap by introducing QRTlib, a library for fast and practical implementations of quantum real transforms, including the quantum Hartley, cosine, and sine transforms of various types. We develop new algorithms and circuit optimizations that make these transforms efficient and suitable for near-term devices. In particular, we present a quantum Hartley transform based on the linear combination of unitaries (LCU) technique, achieving a $4\times$ reduction in circuit size compared to prior methods. We also implement an improved quantum sine transform of Type I that removes the need for large multi-controlled operations. QRTlib provides the first complete implementations of these quantum real transforms in Qiskit.

\paragraph{keywords}
Quantum Hartley Transform, Quantum Cosine Transform, Quantum Sine Transform, Quantum Real Transforms.

\section{Introduction}

Classical discrete transforms such as the discrete cosine transform (DCT), discrete sine transform (DST), and the Hartley transform are central tools in applied mathematics and engineering. The DCT, in particular, is the basis of widely used compression standards such as JPEG for images and MPEG for audio and video, where it enables efficient storage and transmission by concentrating energy into a small number of coefficients \cite{ahmed2006discrete, pennebaker1992jpeg, rao2014discrete, rao2018transform}. The DST has found applications in solving partial differential equations and in spectral methods where boundary conditions make sine expansions more natural than Fourier ones \cite{canuto2006spectral}. The Hartley transform, introduced as a real-valued alternative to the discrete Fourier transform (DFT), has been used extensively in signal and image processing tasks where the avoidance of complex numbers reduces overhead in implementation \cite{bracewell1983discrete, rao2018transform}. These transforms illustrate that, beyond the Fourier transform, real-valued transforms provide indispensable computational primitives in classical computing.

In the quantum setting, the discrete Fourier transform admits an efficient quantum analogue, the quantum Fourier transform ($\qft$) \cite{cleve1998quantum, cleve2000fast}. The $\qft$ has become a cornerstone of quantum computing, enabling powerful algorithms such as Shor’s factoring algorithm and playing a role in algorithms for hidden subgroup problems, phase estimation, and quantum simulation \cite{shor1994algorithms, kaye2006introduction}. By comparison, the study of quantum versions of real transforms such as the cosine, sine, and Hartley transforms is at a much earlier stage. Nevertheless, the classical analogy strongly suggests their potential importance: just as real transforms complement the Fourier transform in classical applications, their quantum counterparts may extend the range of quantum algorithms and provide structural or efficiency advantages. For example, quantum versions of real transforms may be useful in quantum signal and image processing, or in algorithms where restricting to real-valued bases reduces circuit complexity. In this sense, quantum real transforms can be viewed as natural companions to the $\qft$, and developing efficient algorithms and implementations for them is a necessary step toward broadening the algorithmic toolbox of quantum computing.

\paragraph{Previous work.}
Existing work on quantum real transforms has been limited. The work of Klappenecker and R\"{o}tteler introduced constructions of quantum cosine, sine, and Hartley transforms on quantum computers, establishing the theoretical feasibility of these transforms \cite{klappenecker2001discrete, klappenecker2001irresistible}. Subsequent works discussed the quantum sine and cosine transforms in the context of applications such as image and signal processing \cite{pang2006quantum, pang2019signal}. The quantum Hartley transform ($\qht$) was considered in Tseng and Hwang \cite{tseng2005quantum} and Agaian and Klappenecker \cite{agaian2002quantum}, with the latter improving the gate complexity of the earlier construction in \cite{klappenecker2001irresistible}. More recently, Doliskani, Mirzaei, and Mousavi \cite{doliskani2025public} proposed a recursive algorithm for the quantum Hartley transform that further improved upon the algorithm of Agaian and Klappenecker.

While these results provide valuable starting points, they leave open two key gaps: first, the absence of implementations of these algorithms in quantum programming frameworks, and second, the reliance on operations that are impractical on current devices, such as large multi-controlled gates, which significantly increase circuit depth and error rates. As a result, quantum real transforms have remained primarily theoretical constructs rather than practical quantum primitives.

\paragraph{Our contributions.}
In this work, we address these gaps and make three main contributions.

\begin{itemize}
    \item New $\qht$ algorithm. We propose a new quantum Hartley transform algorithm based on the linear combination of unitaries (LCU) technique. Our algorithm employs the quantum Fourier transform together with a simple circuit for a subroutine known as oblivious amplitude amplification. The LCU-based construction provides a more direct and efficient implementation, achieving an approximately fourfold reduction in circuit size compared to the best known algorithm~\cite{doliskani2025public}.

    \item Complete and optimized implementation of the Type-I quantum sine transform. We give a complete implementation of the QFT-based algorithm for $\qst^\mathrm{I}$ whose main idea appears in \cite{cui2021optimization}. We provide the missing operations, a complete derivation, and a proof of correctness. The resulting circuit eliminates the large multi-controlled operator required in the construction of \cite{klappenecker2001discrete} and achieves gate complexity $\frac{1}{2}\log^2 N + O(\log N)$.
    
    \item First implementations. To the best of our knowledge, our library presents the first implementations of the quantum real transforms discussed in this paper, including the Hartley transform and the Type I, II, III, and IV cosine and sine transforms. Our implementations, done using Qiskit and tested in practice, represent the first time these transforms have been made available as concrete quantum circuits rather than purely theoretical designs.
\end{itemize}

Overall, this work provides the first systematic implementation and optimization of quantum real transforms. It shows that these transforms can be realized efficiently and that the resulting circuits are practical for near-term quantum hardware. Furthermore, it lays the groundwork for a comprehensive library of quantum real transforms, similar to the real transform libraries widely used in classical computing.

\paragraph{Code availability.}
The source code for QRTlib is located at \url{https://github.com/jake-doliskani/QRTlib}.

\section{Preliminaries}

\paragraph{Notation.}
In this paper, we always assume the dimension of the ambient Hilbert space is $N = 2^n$, for some integer $n \geq 1$, which represents the state space of an $n$-qubit system. The tensor product $\ket{\psi} \otimes \ket{\phi}$ of two quantum states $\ket{\psi}$ and $\ket{\phi}$ will often be denoted by $\ket{\psi}\ket{\phi}$. The state $\ket{\psi}^{\otimes m}$ (resp. operator $A^{\otimes m}$) denotes the $m$-fold tensor product of the state $\ket{\psi}$ (resp. operator $A$). We define $\omega_m = e^{2\pi i / m}$ as a primitive $m$-th root of unity. 

We denote the conditional one's complement and the conditional two's complement unitaries by $P_{1C}$ and $P_{2C}$, respectively. More precisely,
\begin{align*}
    P_{1C} & = \ket{0}\bra{0} \otimes \mathds{1}_N + \ket{1}\bra{1} \otimes \ket{N - x - 1}\bra{x} \\
    P_{2C} & = \ket{0}\bra{0} \otimes \mathds{1}_N + \ket{1}\bra{1} \otimes \ket{(N - x) \bmod{N}}\bra{x}
\end{align*}
The conditional modular increment-by-one and decrement-by-one unitaries for an $n$-qubit input are denoted by $\inc_n$ and $\dec_n$, respectively, more precisely,
\begin{align*}
    \inc_n & = \ket{0}\bra{0} \otimes \mathds{1}_N + \ket{1}\bra{1} \otimes \ket{(x + 1) \bmod{N}}\bra{x} \\
    \dec_n & = \ket{0}\bra{0} \otimes \mathds{1}_N + \ket{1}\bra{1} \otimes \ket{(x - 1) \bmod{N}}\bra{x}
\end{align*}

\paragraph{The Fourier transform.}
Let $\Z_N$ be the group of integers mod $N$. The Fourier transform of a function $f: \Z_N \to \C$ is given by
\[ \cft_N(f)(a) = \frac{1}{\sqrt{N}} \sum_{y = 0}^{N - 1} \omega_N^{ay} f(y). \]
The quantum Fourier transform of a quantum state $\ket{\psi} = \sum_{x \in \Z_N} f(x) \ket{x}$ is given by $\qft_N\ket{\psi} = \sum_{y \in \Z_N} \cft_N(y) \ket{y}$. For a basis state $\ket{a}$, where $a \in \Z_N$, we have
\[ \qft_N \ket{a} = \frac{1}{\sqrt{N}} \sum_{y = 0}^{N - 1} \omega_N^{ay} \ket{y}. \]

\paragraph{The Hartley transform.}
The Hartley transform of a function $f: \Z_N \to \R$ is the function $\cht_N(f): \Z_N \to \R$ defined by
\[ \cht_N(f)(a) = \frac{1}{\sqrt{N}} \sum_{y = 0}^{N - 1} \cas\Big(\frac{2 \pi ay}{N}\Big) f(y),  \]
where $\cas(x) = \cos(x) + \sin(x)$. Like the Fourier transform, $\cht_N$ is a linear operator and it is unitary. It follows from the identity
\[ \cas\Big(\frac{2 \pi x}{N}\Big) = \frac{1 - i}{2} \omega_N^x + \frac{1 + i}{2} \omega_N^{-x} \]
that
\begin{equation}
    \label{eq:ht-ft}
    \cht_N = \frac{1 - i}{2}\cft_N + \frac{1 + i}{2}\cft_N^*.
\end{equation}
The quantum Hartley transform of a basis state $\ket{a}$ is given by
\[ \qht_N\ket{a} = \frac{1}{\sqrt{N}} \sum_{y = 0}^{N - 1} \cas\Big(\frac{2 \pi ay}{N}\Big) \ket{y}. \]

\paragraph{Cosine and sine transforms.}
The different types of the discrete sine transforms, which we consider in this paper, are
\begin{align*}
    S_N^\text{I} & = \sqrt{\frac{2}{N}} \Big[ \sin\Big( \frac{m n \pi}{N} \Big)  \Big], \, 1 \leq m, n < N \\
    S_N^\text{II} & = \sqrt{\frac{2}{N}} \bigg[ k_m \sin\bigg( \frac{m (n + \frac{1}{2}) \pi}{N} \bigg)  \bigg], \, 1 \leq m, n \leq N \\    
    S_N^\text{III} & = \sqrt{\frac{2}{N}} \bigg[ k_n \sin\bigg( \frac{(m + \frac{1}{2}) n \pi}{N} \bigg)  \bigg], \, 1 \leq m, n \leq N \\
    S_N^\text{IV} & = \sqrt{\frac{2}{N}} \bigg[ \sin\bigg( \frac{(m + \frac{1}{2}) (n + \frac{1}{2}) \pi}{N} \bigg)  \bigg], \, 0 \leq m, n < N,
\end{align*}
where $k_j = 1 / \sqrt{2}$ for $j = N$ and $k_j = 1$ for $j \ne N$. Different types of the discrete cosine transform are
\begin{align*}
    C_N^\text{I} & = \sqrt{\frac{2}{N}} \Big[ k_m k_n \cos\Big( \frac{m n \pi}{N} \Big)  \Big], \, 0 \leq m, n \leq N \\
    C_N^\text{II} & = \sqrt{\frac{2}{N}} \bigg[ k_m \cos\bigg( \frac{m (n + \frac{1}{2}) \pi}{N} \bigg)  \bigg], \, 0 \leq m, n < N \\    
    C_N^\text{III} & = \sqrt{\frac{2}{N}} \bigg[ k_n \cos\bigg( \frac{(m + \frac{1}{2}) n \pi}{N} \bigg)  \bigg], \, 0 \leq m, n < N \\    
    C_N^\text{IV} & = \sqrt{\frac{2}{N}} \bigg[ \cos\bigg( \frac{(m + \frac{1}{2}) (n + \frac{1}{2}) \pi}{N} \bigg)  \bigg], \, 0 \leq m, n < N.
\end{align*}
The quantum versions of these transforms are denoted by $\qst_N^x$, for the quantum sine transform, and $\qct_N^x$ for quantum cosine transform of type $x \in \{\text{I, II, III, IV} \}$. For example, for a given basis element $\ket{a}$,
\[ \qst_N^\text{I} \ket{a} = \sqrt{\frac{2}{N}} \sum_{y = 1}^{N - 1} \sin\left( \frac{\pi a y}{N} \right) \ket{y}, \]
and 
\[ \qct_N^\text{I} \ket{a} = \sqrt{\frac{2}{N}} \sum_{y = 0}^{N} k_a k_y \cos\left( \frac{\pi a y}{N} \right) \ket{y}. \]

\section{Fast Quantum Hartley Transform}

The best-known algorithm for the Quantum Hartley Transform $\qht_N$ is the recursive algorithm proposed by Doliskani, Mirzaei, and Mousavi \cite{doliskani2025public}. In this section, we propose a new algorithm based on the Linear Combination of Unitaries (LCU) technique. For completeness, we first outline the recursive algorithm of \cite{doliskani2025public} and highlight some optimizations for its implementation.

\subsection{The recursive approach}
The core idea of the recursive approach is to exploit a structure-preserving identity that separates the most significant qubit and expresses the $n$-qubit transform $\qht_N$ in terms of the $(n-1)$-qubit transform $\qht_{N/2}$ together with some elementary operations.

Rewriting the action of $\qht_N$ of an $n$-qubit basis element $\ket{y}$, we obtain
\begin{align}
	& \qht_N\ket{y} \nonumber \\
    & = \frac{1}{\sqrt{N}} \sum_{y = 0}^{N - 1} \cas\Big( \frac{2\pi a y}{N} \Big) \ket{y} \nonumber \\
    & = \frac{1}{\sqrt{N}} \sum_{y = 0}^{N / 2 - 1} \cas\Big( \frac{2\pi a y}{N} \Big) (\ket{y} + (-1)^a \ket{y + N/2}) \nonumber \\
    & = \sqrt{\frac{2}{N}} \sum_{y = 0}^{N / 2 - 1} \cas\Big( \frac{2\pi a y}{N} \Big) \frac{1}{\sqrt{2}} (\ket{0} + (-1)^a \ket{1}) \ket{y}, \label{eq:qht-alt} 
\end{align}
where in the last equality, we have decomposed the system into a register containing the first qubit and a register containing the remaining $n-1$ qubits. This decomposition motivates the recursive implementation outlined in Algorithm~\ref{alg:qht-rec}. The algorithm computes $\qht_N$ by augmenting the state with a single ancilla qubit, applying $\qht_{N/2}$ to the reduced register, and using elementary gates to construct the final output. The algorithm requires $2\log^2 N + O(\log N)$ elementary gates to implement \cite[Theorem 4.1]{doliskani2025public}.

We briefly explain the transforms used in the algorithm. The conditional operator 
\[ \ket{0}\ket{y} \mapsto \ket{0}\ket{y}, \ket{1}\ket{y} \mapsto \ket{1}\ket{N/2 - y}, \]
mentioned in \cite{doliskani2025public}, is the conditional two's complement unitary $P_{2C}$, i.e., the value of the second register is $N / 2 - y \bmod{N / 2}$. The unitary $U_R$ is defined by the action $U_R \ket{c}\ket{y}\ket{b} = (R(y, b)\ket{c})\ket{y}\ket{b}$, where $R(y, b)$ is a single-qubit rotation defined by
\[
    R(y, b) = 
    \begin{bmatrix}
        \cos(2\pi b y / N) & \sin(2\pi b y / N) \\
        -\sin(2\pi b y / N) & \cos(2\pi b y / N)
    \end{bmatrix}.
\]
The unitary $C_X$ is a multi controlled-\cnot{} that takes $\ket{c} \ket{y} \ket{b}$ to $\ket{c} \ket{y} \ket{1 \oplus b}$ if $c = 1$ and $y = 0$, and acts as identity otherwise.

\begin{algorithm}
    \caption{Recursive $\qht_N$}
    \label{alg:qht-rec}
    \begin{algorithmic}[1]
        \Require $n$-qubit state $\ket{\psi}$
        \Ensure $\qht_N \ket{\psi}$.
        
        \State Initialize an ancilla qubit to $0$ to obtain the state $\ket{0}\ket{\psi}$
        \State Compute $\mathds{1}_2 \otimes \qht_{N/2} \otimes \mathds{1}_2$ recursively
        \State Apply $H \otimes \mathds{1}_N$
        \State\label{stp:twos-cmp}Perform the conditional two's complement on the first $n$ qubits, using the first qubit as control. \Comment{the unitary $P_{2C}$}
        \State Apply the unitary $U_R$
        \State Perform the conditional two's complement on the first $n$ qubits, using the first qubit as control. \Comment{the unitary $P_{2C}$}
        \State Apply the unitary $(H \otimes \mathds{1}_N)C_X(H \otimes \mathds{1}_N)$
        \State Apply $H \otimes \mathds{1}_N$
        \State Apply $\cnot{}$ to the first and last qubits
        \State Apply $\mathds{1}_N \otimes H$
        \State Relabel the qubits to implement the effect of the swap $\ket{0}\ket{y}\ket{b} \mapsto \ket{0}\ket{b}\ket{y}$
        \State Trace out the first qubit
    \end{algorithmic}
\end{algorithm}

\subsection{Implementation of an optimized two's complement operation}

One of the expensive steps of Algorithm \ref{alg:qht-rec}, despite its apparent simplicity, is the conditional two's complement operation in Step \ref{stp:twos-cmp}. The two's complement operation is a standard arithmetic primitive and efficient quantum circuits for it can be obtained from known quantum addition circuits. We do not claim a new two's complement algorithm here. However, its implementation has a significant effect on the resource requirements of Algorithm \ref{alg:qht-rec}, and we therefore describe the particular construction used in our library and its gate count in some detail.

There are two natural approaches to implementing this operation: without an ancilla register, and with an ancilla register. An ancilla-free implementation can be obtained using multi-controlled \cnot{} gates with an increasing number of controls. For an $n$-qubit register, this requires gates with up to $n$ controls. Decomposing these large multi-controlled gates into elementary gates is costly, increasing the gate complexity of the two's complement operation to $O(\log^2 N)$ and the overall complexity of Algorithm \ref{alg:qht-rec} to $O(\log^3 N)$. Moreover, practical decompositions of such multi-controlled gates often introduce ancillary qubits themselves.

Instead, our implementation uses an ancilla-based construction with $O(\log N)$ gates. It follows the standard two-step realization of the two's complement: first negate all data qubits and then increment the resulting value by $1$. The main nontrivial component is therefore the increment-by-one operation. We implement this operation using a ripple-carry construction adapted from the constant-adder circuit of \cite{vedral1996quantum}. Although the underlying construction is standard, we include the details because the resulting circuit is used repeatedly in our implementations and its precise resource requirements are important for the gate-complexity estimates.

The increment circuit introduces a sequence of ancillary carry qubits that temporarily store the intermediate carries. For an $n$-qubit data register, exactly $n - 2$ carry qubits are sufficient. Let $b_0,\ldots,b_{n-1}$ denote the data qubits, ordered from least significant to most significant, and let $a_1,\ldots,a_{n-2}$ denote the carry qubits.

The circuit first computes the carry $a_1$ from the two least significant data qubits $b_0$ and $b_1$. Each subsequent carry $a_{i+1}$ is then computed from the preceding carry $a_i$ and the next data qubit $b_{i+1}$. This forward pass propagates the carry information toward the most significant end of the register.

Once the carry values have been computed, the data qubits are updated in reverse order, from most significant to least significant, according to the corresponding carry values. At the same time, the carry qubits are uncomputed by reversing the carry-propagation procedure. Consequently, all ancillary qubits are returned to the state $\ket{0}$ at the end of the computation, leaving only the incremented value in the data register. Figure \ref{fig:inc_1} shows the resulting circuit for a $5$-qubit data register with $3$ carry qubits.

\begin{figure}[t]
    \includegraphics[width = \columnwidth]{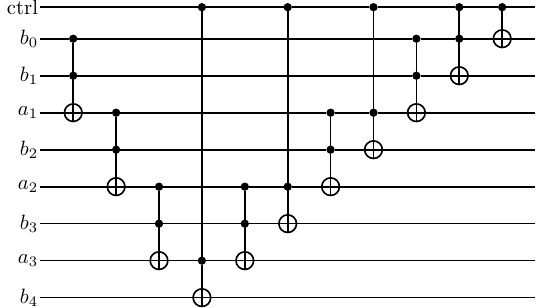}
    \caption{The circuit for $\inc\_5$. The $b_i$ and $a_i$ represent the data and carry qubits, respectively.}
    \label{fig:inc_1}
\end{figure}

The upper portion of the circuit performs the forward carry propagation, while the lower portion updates the data qubits and uncomputes the carry register. In our circuit model, the resulting conditional two's complement operation on an $n$-qubit data register uses $n - 2$ ancillary qubits and $4n - 4$ gates. Thus, since $n = \log N$, the operation has $O(\log N)$ gate complexity. Figure \ref{fig:snip-inc_1} shows a portion of our implementation of the carry-propagation and uncomputation steps in the $\inc_n$ routine.

\begin{figure}
\begin{minted}{python}

# Forward propagate carries: a_{i+1} = a_i AND b_{i+1}

for i in range(len(anc_qubits) - 1):
    circuit.ccx(
        anc_qubits[i],
        target_qubits[i + 2],
        anc_qubits[i + 1]
    )

# Backward pass (ripple under ctrl):

for i in range(len(anc_qubits) - 1, 0, -1):
    circuit.ccx(
        ctrl,
        anc_qubits[i],
        target_qubits[i + 2]
    )  # flip b_{i+2} if ctrl=1 and a_i=1
    circuit.ccx(
        anc_qubits[i - 1],
        target_qubits[i + 1],
        anc_qubits[i]
    )  # uncompute a_i
\end{minted}
    \caption{Carry propagation and uncomputation in the implementation of $\inc_n$.}
    \label{fig:snip-inc_1}
\end{figure}

\subsection{A new $\qht$ algorithm using LCU}

In this section, we propose a new $\qht_N$ algorithm that utilizes a well-known technique called Linear Combination of Unitaries (LCU). We refer the reader to Appendix \ref{sec:LCU} for a brief overview of the LCU technique. Let $T: \Z_N \to \Z_N$ be the two's complement unitary, that acts on any function $f: \Z_N \to \R$ by $Tf(x) = f(N - x \bmod N)$. Then we have
\begin{align*}
    (F_NTf)(x)
    & = \frac{1}{\sqrt{N}} \sum_{y = 0}^{N - 1} \omega_N^{xy} f(N - y \bmod N) \\
    & = \frac{1}{\sqrt{N}} \sum_{y = 0}^{N - 1} \omega_N^{-xy} f(y) \\
    & = (F_N^*f)(x).
\end{align*}
Therefore, $F_NT = F_N^*$. Using this, we can rewrite the Hartley transform $H_N$ as
\begin{align*}
    H_N
    &= \Re(F_N) + \Im(F_N) \\
    &= \frac{1}{2}(F_N + F_N T) + \frac{1}{2i}(F_N - F_N T) \\
    &= F_N \Big( \frac{1}{2}(1 + \frac{1}{i}) \mathds{1}_N + \frac{1}{2}(1 - \frac{1}{i}) T \Big) \\
    &= F_N \Big( \frac{1}{\sqrt{2}} e^{-i\frac{\pi}{4}} \mathds{1}_N + \frac{1}{\sqrt{2}} e^{i\frac{\pi}{4}} T \Big)
\end{align*}
We now show how to implement $\qht_N$ via the LCU framework. Let $V = (e^{-i\frac{\pi}{4}} \mathds{1}_N + e^{i\frac{\pi}{4}} T) / \sqrt{2}$. Absorbing the complex coefficients into the unitaries as global phase gates, we define 
\[
    U_0 = e^{-i\frac{\pi}{4}} \mathds{1}_N, \quad U_1 = e^{i\frac{\pi}{4}} T.
\]
Then
\begin{equation}
    \label{eq:lcu-V}
    V = \frac{1}{\sqrt{2}} U_0 + \frac{1}{\sqrt{2}} U_1, \quad \qht_N = \qft_N V,
\end{equation}
and implementing $V$ would result in an efficient algorithm for $\qht_N$. From the definition of $V$, we obtain the LCU setting (Appendix \ref{sec:LCU}) as follows. We have $m = 1$, so $M = 2^m = 2$, and $a_1 = a_2 = 1/\sqrt{2}$. Therefore, the operator $A$ in \eqref{eq:luc-A} satisfies $A \ket{0} = (\ket{0} + \ket{1}) / \sqrt{2}$, which means $A = H$. The unitary $U$ is $U = \ket{0}\bra{0} \otimes U_0 + \ket{1}\bra{1} \otimes U_1$. The unitary $W$ is
\[
    W = (H \otimes \mathds{1}_N) U (H \otimes \mathds{1}_N).
\]
The reflection $R$ is $R = 2\ket{0}\bra{0} - \mathds{1}$ on the first qubit, which is just the Pauli $Z$ operator. Finally, the oblivious amplitude amplification operator is $\mathcal{S} = -W R W^* R$. The action of $W$ on the input state $\ket{0} \ket{\psi}$ is
\[
    W \ket{0} \ket{\psi} = \sin \theta \ket{0} V\ket{\psi} + \cos \theta \ket{\phi^\perp}, \quad \text{where } \theta = \frac{\pi}{4}.
\]
Since $\pi/(2\theta) = 2$ is not an odd integer, we cannot use Lemma \ref{lem:obl-amp} to obtain the perfect transform $\ket{0}\ket{\psi} \mapsto \ket{0} V\ket{\psi}$. Instead, we use the angle $\theta' = \pi/6$ to obtain the operator $P$ defined by
\begin{align*}
    P: \ket{0}
    & \mapsto \Big( \frac{\sin\theta'}{\sin\theta} \Big) \ket{0} + \sqrt{1 - \Big( \frac{\sin\theta'}{\sin\theta} \Big)^2} \ket{1} \\
    & = \frac{1}{\sqrt{2}} (\ket{0} + \ket{1}), 
\end{align*}
i.e., $P = H$. We then define $W' = P \otimes W = H \otimes W$ and use a 2-qubit ancilla register to obtain
\[
    W' \ket{00} \ket{\psi} = \sin \theta' \ket{00} V\ket{\psi} + \cos \theta' \ket{\phi^\perp}.
\]
The new oblivious amplitude amplification operator becomes
\[
    \mathcal{S}' = -W' R' (W')^* R',
\]
where $R' = 2\ket{00}\bra{00} - \mathds{1}$. The final state after one round of amplification is
\[
    \mathcal{S'} W' \ket{00} \ket{\psi} = \ket{00} V \ket{\psi}.
\]
We have outlined the implementation of $\qht_N$ using the above LCU procedure in Algorithm \ref{alg:qht-lcu}. 

\begin{algorithm}
    \caption{$\qht_N$ via LCU}
    \label{alg:qht-lcu}
    \begin{algorithmic}[1]
        \Require $n$-qubit state $\ket{\psi}$.
        \Ensure $\qht_N\ket{\psi}$

        \State Prepare the state $\ket{00}\ket{\psi}$ by appending a 2-qubit ancilla in the zero state.
        \State Apply the unitary $W'$ to the state $\ket{00}\ket{\psi}$.
        \State Apply $R'$, $W'^*$, $R'$ and $-W'$ in that order.
        \State Trace out the first two qubits.
        \State Apply $\qft_N$ to the remaining register.
    \end{algorithmic}
\end{algorithm}

\begin{theorem}
    Algorithm \ref{alg:qht-lcu} correctly implements the quantum Hartley transform on $\qht_N$ using $\frac{1}{2} \log^2 N + O(\log N)$ elementary gates.
\end{theorem}
\begin{proof}
    The correctness of the algorithm follows from the preceding discussion. Let us now analyze its gate complexity. The unitaries $U_0$ and $U_1$ are implemented using elementary phase gates and a two 's-complement gate. Consequently, the unitary $W$, and hence $W'$, can be implemented using $O(\log N)$ elementary gates. The reflection $R'$ can be implemented using $O(1)$ elementary gates. Therefore, the unitary $V$ in \eqref{eq:lcu-V} can be implemented using $O(\log N)$ elementary gates.

    Since the gate complexity of $\qft_N$ is $\tfrac{1}{2} \log^2 N + O(\log N)$, the overall gate complexity of $\qht_N$ is also $\tfrac{1}{2} \log^2 N + O(\log N)$.
\end{proof}

\noindent Figure \ref{fig:unitary-w} is a snippet from our implementation of $\qht_N$, where we implement the unitary $W$.
\begin{figure}
\begin{minted}{python}
def _build_unitary_w(data: list[int]):
    n = len(data)

    qc = QuantumCircuit(2 * n - 1, name="w")
    control_qubit = 0
    data_qubits = list(range(1, n + 1))
    anc_qubits = list(range(n + 1, 2 * n - 1))

    qc.h(control_qubit)
    ctrl_twos_complement(
              qc, 
              anc_qubits[0 : n - 2], 
              data_qubits + [control_qubit])
    qc.rz(np.pi / 2, control_qubit)
    qc.h(control_qubit)

    return qc.to_gate(label="W")
\end{minted}
    \caption{Unitary $W$ for the LCU subroutine}
    \label{fig:unitary-w}
\end{figure}

\noindent The gate complexity for various $\qht_N$ algorithms are compared in Table \ref{tbl:hartley-cmp}.
\begin{table}
    \centering
    \footnotesize
    \renewcommand{\arraystretch}{1.4}
    \renewcommand\tabularxcolumn[1]{m{#1}}
    \begin{tabularx}{\columnwidth}{|X|l|}
        \hline
        $\bm{\qht_N}$ \textbf{algorithm} & \textbf{Gate complexity} \\
        \hline
        Klappenecker and R\"{o}tteler \cite{klappenecker2001irresistible} \newline (Using controlled-$\qft_N$) & $\frac{5}{2}\log^2N + O(\log N)$ \\
        \hline
        Agaian and Klappenecker \cite{agaian2002quantum} \newline (Recursive decomposition) & $\frac{5}{2}\log^2N + O(\log N)$ \\
        \hline
        Doliskani, Mirzaei, and Mousavi \cite{doliskani2025public} \newline (Recursive decomposition) & $2\log^2N + O(\log N)$ \\
        \hline
        This work (using LCU) & $\frac{1}{2}\log^2N + O(\log N)$ \\
        \hline
    \end{tabularx}
    \vspace*{3mm}
    \caption{Gate complexity of different $\qht_N$ algorithms}
    \label{tbl:hartley-cmp}
\end{table}

\subsection{Noise Robustness of the QHT Implementations}
\label{sec:qht_noise}

In this section, we compare the LCU and recursive implementations of
$\qht_N$ under depolarizing noise for each $n = 3, \ldots, 11$ data
qubits.

\paragraph{Noise model.}
The noise simulations are performed using the
\texttt{NoiseModel} class of Qiskit Aer. Before applying noise,
both implementations are transpiled to the same basis
$\{U,\cnot\}$. Here, $U(\theta,\phi,\lambda)$ denotes
the generic single-qubit unitary
\begin{equation}
    U(\theta,\phi,\lambda)
    =
    \begin{bmatrix}
        \cos(\theta/2)
        &
        -e^{i\lambda}\sin(\theta/2)
        \\
        e^{i\phi}\sin(\theta/2)
        &
        e^{i(\phi+\lambda)}\cos(\theta/2)
    \end{bmatrix}.
\end{equation}
Consequently, higher-level operations appearing in the original
$\qht_N$ circuits, such as multi-controlled $X$ gates,
controlled rotations, and the $\qft$, are decomposed into this
common one- and two-qubit basis before noise is introduced.

For a $k$-qubit state $\rho$, the depolarizing channel used by
Qiskit Aer is
\begin{equation}
    \mathcal{D}^{(k)}_{p}(\rho)
    =
    (1-p)\rho
    +
    p\,\operatorname{Tr}(\rho)
    \frac{I_{2^k}}{2^k},
    \label{eq:depolarizing_channel}
\end{equation}
where $p$ is the depolarizing-noise parameter. Since
$\operatorname{Tr}(\rho)=1$ for a normalized density matrix,
the one-qubit and two-qubit channels used in our simulations are
respectively
\begin{align}
    \mathcal{D}^{(1)}_{p}(\rho)
    &= (1-p)\rho + p\frac{I_2}{2}, \\
    \mathcal{D}^{(2)}_{p}(\rho)
    &=
    (1-p)\rho
    +
    p\frac{I_4}{4}.
\end{align}
In our implementation, \texttt{depolarizing\_error($p$,1)} is attached to every
transpiled $U$ gate while
\texttt{depolarizing\_error($p$,2)} is attached to every
transpiled $\cnot$ gate. Therefore, $p$ represents a
per-gate depolarizing-noise strength rather than a single error
applied once to the entire circuit. Every occurrence of a
one- or two-qubit basis gate is followed by its corresponding
noise channel. The values
$p\in\{0,0.001,0.0025,0.005\}$ are considered.

Although the complete circuit, including all ancillary qubits,
is subject to noise, the final performance is evaluated on the
$n$ data qubits. The ancillary subsystem is traced out, and the
resulting noisy data-register density matrix is compared with
the corresponding ideal data-register state using state
fidelity. Thus, the LCU and recursive implementations are
evaluated on the same $n$-qubit logical output space despite
using different numbers of ancillary qubits.

Figure~\ref{fig:qht_noise_3d} shows the resulting fidelity as a
function of both the number of data qubits and the depolarizing
noise strength, while Table~\ref{tab:qht_noise_fidelity}
reports the corresponding numerical values. At $p=0$, both
implementations reproduce the ideal result with unit fidelity.
For three data qubits, their noisy performance is nearly
identical. From four data qubits onward, the LCU implementation of
$\qht_N$ consistently retains higher fidelity
for all tested nonzero noise strengths, and the separation
generally becomes more pronounced as the transform size
increases.

\begin{figure}[t!]
    \centering
    \includegraphics[width=\columnwidth]{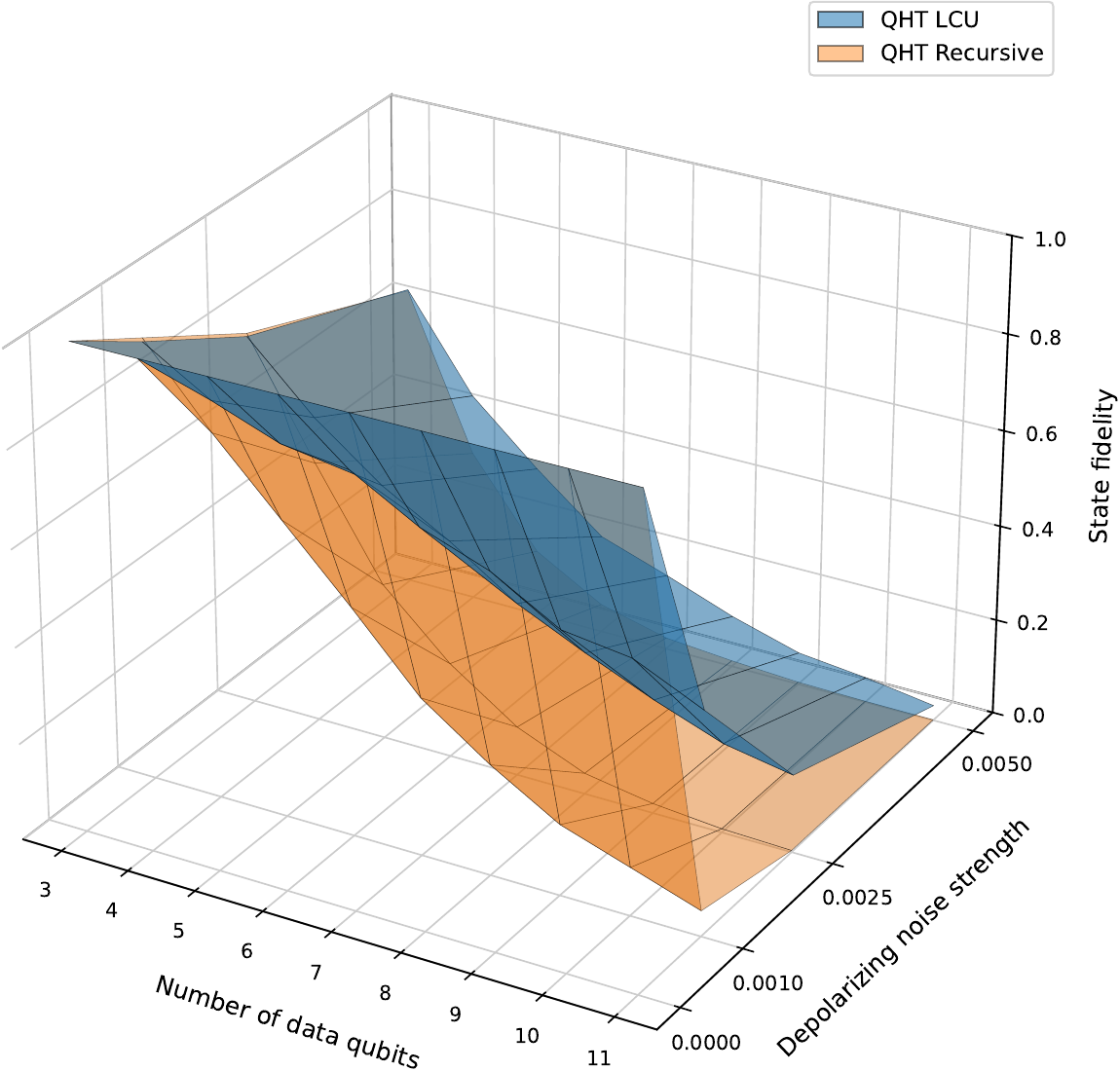}
    \caption{
        Noise robustness of the LCU and recursive implementations of
        $\qht_N$. The surfaces show the data-register state fidelity
        as a function of the number of data qubits and the
        depolarizing-noise parameter $p$.
    }
    \label{fig:qht_noise_3d}
\end{figure}

\begin{table*}[t!]
    \centering
    \small
    \renewcommand{\arraystretch}{1.15}
    \setlength{\tabcolsep}{2pt}

    \begin{tabularx}{\textwidth}{
        @{}
        c
        *{10}{>{\centering\arraybackslash}X}
        @{}
    }

        \toprule

        &
        \multicolumn{5}{c}{
            \textbf{$\qht_N$ LCU}
        }
        &
        \multicolumn{5}{c}{
            \textbf{$\qht_N$ REC}
        }
        \\

        \cmidrule(lr){2-6}
        \cmidrule(lr){7-11}

        \textbf{Data}
        &
        \textbf{Total}
        &
        \multicolumn{4}{c}{\textbf{Fidelity}}
        &
        \textbf{Total}
        &
        \multicolumn{4}{c}{\textbf{Fidelity}}
        \\

        \textbf{qubits}
        &
        \textbf{qubits}
        &
        $p=0$
        &
        $p=0.001$
        &
        $p=0.0025$
        &
        $p=0.005$
        &
        \textbf{qubits}
        &
        $p=0$
        &
        $p=0.001$
        &
        $p=0.0025$
        &
        $p=0.005$
        \\

        \midrule

        3
        & 6
        & 1.0000 & 0.9042 & 0.7729 & 0.6377
        & 4
        & 1.0000 & 0.9120 & 0.7796 & 0.6355
        \\

        4
        & 8
        & 1.0000 & 0.8210 & 0.6401 & 0.4429
        & 6
        & 1.0000 & 0.7559 & 0.5456 & 0.3178
        \\

        5
        & 10
        & 1.0000 & 0.7701 & 0.5352 & 0.3198
        & 8
        & 1.0000 & 0.6161 & 0.3272 & 0.1394
        \\

        6
        & 12
        & 1.0000 & 0.7536 & 0.4562 & 0.2085
        & 10
        & 1.0000 & 0.4757 & 0.1978 & 0.0593
        \\

        7
        & 14
        & 1.0000 & 0.6755 & 0.3629 & 0.1608
        & 12
        & 1.0000 & 0.3285 & 0.1021 & 0.0239
        \\

        8
        & 16
        & 1.0000 & 0.6249 & 0.2983 & 0.1092
        & 14
        & 1.0000 & 0.2302 & 0.0437 & 0.0122
        \\

        9
        & 18
        & 1.0000 & 0.5451 & 0.2390 & 0.0682
        & 16
        & 1.0000 & 0.1475 & 0.0201 & 0.0043
        \\

        10
        & 20
        & 1.0000 & 0.5268 & 0.1876 & 0.0531
        & 18
        & 1.0000 & 0.1006 & 0.0077 & 0.0018
        \\

        11
        & 22
        & 1.0000 & 0.4618 & 0.1619 & 0.0319
        & 20
        & 1.0000 & 0.0532 & 0.0035 & 0.0010
        \\

        \bottomrule

    \end{tabularx}
    \vspace*{3mm}
    \caption{
        State fidelity of $\qht_N$ under different depolarizing-noise
        strengths.
    }
    \label{tab:qht_noise_fidelity}
\end{table*}

\section{Type-I Quantum Cosine and Sine Transforms}
\label{sec:simult-qcst-I}

In this section, we describe the implementations for the Type-I Quantum Cosine and Sine Transforms. The first implementation follows a method that enables the simultaneous computation of both $\qct_N^{\mathrm{I}}$ and $\qst_N^{\mathrm{I}}$ using a single circuit \cite{klappenecker2001discrete}. The second implementation is a customized optimization that simplifies the structure and targets only the $\qst_N^{\mathrm{I}}$ output \cite{doliskani2025public}. Both implementations have been tested and included in the library.

\subsection{Simultaneous since-cosine transform}

The circuit for Type-I Quantum Cosine and Sine Transforms, presented in \cite{klappenecker2001discrete}, is constructed using the following key matrix identity:
\[
    T_N^* \cdot F_{2N} \cdot T_N = C_N^{\text{I}} \oplus i S_N^{\text{I}},
\]
where \(C_N^{\text{I}}\) and \(S_N^{\text{I}}\) denote the Type-I cosine and sine transforms, respectively. The direct sum operator \(\oplus\) indicates that the transform output is split into two orthogonal subspaces: the cosine transform is applied when the control qubit is in the \(\ket{0}\) state, and the sine transform is applied when the control qubit is in the \(\ket{1}\) state. In our notation, this is equivalent to
\begin{equation}
    \label{eq:qst-qct-I}
    T_N^* \cdot \qft_{2N} \cdot T_N = \ket{0}\bra{0} \otimes \qct_N^{\text{I}} + i \ket{1}\bra{1} \otimes \qst_N^{\text{I}}.
\end{equation}
Therefore, to apply $\qct_N^{\text{I}}$ to a given state $\ket{\psi}$, one should append a single ancilla to prepare the state $\ket{0}\ket{\psi}$ and apply the above unitary. Similarly, to apply $\qst_N^{\text{I}}$, one starts with $\ket{1}\ket{\psi}$ and then clear out the phase $i$ using an $S$-gate.

The unitary transform $T_N$ is defined by the following action:
\begin{align*}
    T_N \ket{00} &= \ket{00}, \\
    T_N \ket{0x} &= \frac{1}{\sqrt{2}} \ket{0x} + \frac{1}{\sqrt{2}} \ket{1x'}, \\
    T_N \ket{10} &= \ket{10}, \\
    T_N \ket{1x} &= \frac{i}{\sqrt{2}} \ket{0x} - \frac{i}{\sqrt{2}} \ket{1x'},
\end{align*}
where \(x \in \{1, \dots, N-1\}\), and \(x'\) denotes the two's complement of \(x\). The unitary \(T_N\) can be decomposed into a product of two unitarys $T_N = P_{2C} \cdot D$, where \(D\) is a conditional rotation gate acting on the control qubit based on the data register, and $P_{2C}$ is conditioned on the control qubit. More precisely, the gate \(D\) is defined by:
\begin{align*}
    D\ket{00} &= \ket{00}, \\
    D\ket{0x} &= \frac{1}{\sqrt{2}} \ket{0x} + \frac{1}{\sqrt{2}} \ket{1x}, \\
    D\ket{10} &= \ket{10}, \\
    D\ket{1x} &= \frac{i}{\sqrt{2}} \ket{0x} - \frac{i}{\sqrt{2}} \ket{1x},
\end{align*}
where \(x \in \{1, \dots, N-1\}\). Therefore, $D$ can be implemented by applying an \(S\) gate followed by a Hadamard gate on the control qubit conditioned on the data register being non-zero. Algorithm \ref{alg:QCT_QST_Type_I} outlines the steps for the implementation of the simultaneous quantum cosine and sine transforms using the unitary in \eqref{eq:qst-qct-I}.

\begin{algorithm}
    \caption{Type-I Quantum Cosine and Sine Transforms}
    \label{alg:QCT_QST_Type_I}
    \begin{algorithmic}[1]
        \Require $(n + 1)$-qubit state $\ket{c}\ket{\psi}$, where $c \in \{0,1\}$.
        \Ensure $\ket{0}\qct_N^{\text{I}}\ket{\psi}$ if $c = 0$, and $\ket{1}\qst_N^{\text{I}}\ket{\psi}$ if $c = 1$.
        
        \State Apply the controlled unitaries $S \otimes \mathds{1}_N$ and $H \otimes \mathds{1}_N$, conditioned on the last $n$-qubits not being in the all-zero state $\ket{0^n}$.
        \State Perform the conditional two's complement using the first qubit as control. \Comment{the unitary $P_{2C}$}
        \State Apply the $\qft_{2N}$
        \State Perform the conditional two's complement using the first qubit as control. \Comment{the unitary $P_{2C}$}
        \State Apply the controlled unitareis $H \otimes \mathds{1}_N$ and $S \otimes \mathds{1}_N$, conditioned on the last $n$-qubits not being in the all-zero state $\ket{0^n}$.
        \State Apply an $S^*$ gate to the first qubit.
    \end{algorithmic}
\end{algorithm}

\subsection{Implementation}
\label{sec:sine-cosine-I-impl}

To implement the conditional rotation in the $D$ gate, we need to determine whether the $n$-qubit data register is nonzero. This amounts to computing the logical \orop{} of the data qubits into an ancilla and using the resulting qubit as the control for the $S$ and $H$ gates. This is a standard reversible operation, and Qiskit provides the \texttt{OrGate} primitive for this purpose. The \texttt{OrGate} is implemented using a multi-open-controlled $X$ gate together with a single-qubit $X$ gate.

For a concrete resource estimate, we use Qiskit's synthesis for the multi-controlled $X$ operation, based on the construction of \cite{khattar2025rise}. For an $n$-qubit data register, this synthesis requires one clean work qubit, in addition to the qubit storing the result of the \orop{} operation, and uses at most $6n - 6$ \cnot{} gates. Since the \orop{} operation is computed before the conditional rotation and uncomputed afterward, the complete nonzero test and ancilla recovery require at most $12n - 12$ \cnot{} gates. In particular, since $n = \log N$, this contributes only $O(\log N)$ gates to the transform.

\subsection{Optimized type-I sine transform}

If one is interested only in the sine transform, a more efficient construction is possible that avoids the large $n$-qubit-controlled gates used in Algorithm \ref{alg:QCT_QST_Type_I}. The basic idea behind this construction already appears in \cite{cui2021optimization}, but the algorithm presented there is incomplete. In particular, the circuit omits the second, adjoint two's complement operation, as well as several elementary operations required to obtain the desired transform, and a complete derivation and correctness proof are not provided. A complete construction based on the same idea was later given by Doliskani, Mirzaei, and Mousavi \cite{doliskani2025public}, using the quantum Hartley transform. Their algorithm has gate complexity $2 \log^2 N + O(\log N)$. In this section, we give a complete version of the construction for $\qst_N^{\mathrm{I}}$ based directly on the quantum Fourier transform, including all required operations and a proof of correctness. Using the quantum Fourier transform reduces the gate complexity to $\frac{1}{2} \log^2 N + O(\log N)$.

Although this new algorithm for $\qst^{\text{I}}_N$ has the same asymptotic gate complexity as Algorithm \ref{alg:QCT_QST_Type_I}, an important optimization in the new algorithm lies in the removal of the large controlled gate structure used to detect whether the data register is non-zero in the algorithm of \cite{klappenecker2001discrete}. This detection was necessary to ensure the conditional application of the \(S\) and \(H\) gates on the control qubit. However, in Type-I Discrete Sine Transform (DST-I), the domain is restricted to indices \(1, 2, \dots, N-1\), so the data register never takes values \(0\) or \(N\) (which would evaluate to \(0 \mod N\)), both of which correspond to zero output amplitude in the sine basis. As a result, the controlled detection circuit can be eliminated entirely. Unfortunately, the same technique does not seem to adapt to the cosine transform in a straightforward way.

The algorithm proceeds as follows. Given a basis state $\ket{a}$, we prepare the state $\ket{0}\ket{a}$ by appending an ancilla qubit in the zero state, which will be used as a control qubit. We then apply an \(X\) gate followed by a Hadamard gate on the control qubit. This transforms the initial state into the superposition
\[
    \frac{1}{\sqrt{2}} \left( \ket{0}\ket{a} - \ket{1}\ket{a} \right).
\]
Applying $P_{2C}$ to the above state, using the first qubit as control, gives
\[
    \frac{1}{\sqrt{2}} \left( \ket{0}\ket{a} - \ket{1}\ket{N - a} \right).
\]
Denote the above operations as $A_N$, i.e., $A_N = P_{2C} (H \otimes \mathds{1}_N)$. Applying a $\qft_{2N}$ to the entire state results in the state
\begin{align*}
    \ket{\psi}
    & = \frac{1}{2\sqrt{N}} \sum_{y=0}^{2N-1} \Big( \omega_{2N}^{ay} - \omega_{2N}^{-ay} \Big) \ket{y} \\
    & = \frac{i}{\sqrt{N}} \sum_{y=1}^{2N-1} \sin\left( \frac{\pi a y}{N} \right) \ket{y}.
\end{align*}
This state can be rewritten, by separating the first and second halves of the sum, and a change of variables, as follows:
\begin{align*}
    \ket{\psi}
    & = \frac{i}{\sqrt{N}} \sum_{y = 1}^{N - 1} \sin\Big( \frac{\pi a y}{N} \Big) \ket{y} + \frac{i}{\sqrt{N}} \sum_{y = N}^{2N - 1} \sin\Big( \frac{\pi a y}{N} \Big) \ket{y} \\
    & = \frac{i}{\sqrt{N}} \sum_{y = 1}^{N - 1} \sin\Big( \frac{\pi a y}{N} \Big) \ket{y} + \\
    & \mathrel{\phantom{=}} \frac{i}{\sqrt{N}} \sum_{y = 1}^{N - 1} \sin\Big( \frac{\pi a (2N - y)}{N} \Big) \ket{2N - y} \\
    & = \frac{i}{\sqrt{N}} \sum_{y = 1}^{N - 1} \sin\Big( \frac{\pi a y}{N} \Big) (\ket{y} - \ket{2N - y})
\end{align*}
Separating the most significant qubit (the control qubit encoded in the most significant bit due to little-endian layout) we obtain
\[
    \ket{\psi} = \sqrt{\frac{2}{N}} \sum_{y=1}^{N-1} \sin\left( \frac{\pi a y}{N} \right) 
    \frac{i}{\sqrt{2}} \left( \ket{0}\ket{y} - \ket{1}\ket{N - y} \right).
\]
Now, we apply the inverse of \(A_N\) to eliminate the entanglement between the control and data registers and collapse the control qubit back to \(\ket{0}\). The resulting state is
\[
    i\ket{1} \sqrt{\frac{2}{N}} \sum_{y=1}^{N-1} \sin\left( \frac{\pi a y}{N} \right) \ket{y}.
\]
To summarize, we have constructed an efficient unitary \(A_N\) such that
\[
    \left( A_N^* \cdot \qft_{2N} \cdot A_N (X \otimes \mathds{1}_N)\right) \ket{0}\ket{a} = i\ket{1} \qst^{\text{I}}_{N-1} \ket{a},
\]
Finally, we apply an $S^*$ gate to eliminate the phase $i$, and an $X$ gate to obtain $\ket{0} \qst^{\text{I}}_{N-1} \ket{a}$. We have outlined the above steps in Algorithm \ref{alg:Optimized_QST}.

\begin{algorithm}[t]
    \caption{Optimized Quantum Sine Transform}
    \label{alg:Optimized_QST}
    \begin{algorithmic}[1]
        \Require $n$-qubit state $\ket{\psi}$.
        \Ensure $\qst^{\text{I}}_N\ket{\psi}$.
        \State Prepare the state $\ket{0}\ket{\psi}$ by appending a 1-qubit ancilla in the zero state. 
        \State Apply $X \otimes \mathds{1}_N$ and $H \otimes \mathds{1}_N$.
        \State Perform the conditional two's complement using the first qubit as control. \Comment{the unitary $P_{2C}$}
        \State Apply the $\qft_{2N}$.
        \State Perform the conditional two's complement using the first qubit as control. \Comment{the unitary $P_{2C}$}
        \State Apply $H \otimes \mathds{1}_N$, $S^* \otimes \mathds{1}_N$ and $X \otimes \mathds{1}_N$.
        \State Trace out the ancilla qubit
    \end{algorithmic}
\end{algorithm}

\begin{theorem}
    Algorithm \ref{alg:Optimized_QST} applies $\qst^{\mathrm{I}}_N$ using $\frac{1}{2}\log^2N + O(\log N)$ elementary gates.
\end{theorem}
\begin{proof}
    The $\qft_{2N}$ has gate complexity $\frac{1}{2}\log^2N + O(\log N)$. All other unitaries in the algorithm can be implemented using $O(\log N)$ elementary gates.
\end{proof}

\noindent The snippet in Figure \ref{fig:unitary-A-impl} shows the implementation of the unitary $A_N$ defined above.

\begin{figure}
\begin{minted}{python}
# control qubit for the transformation
ctrl = target_qubits[-1]  
circuit.x(ctrl)  # X on control

# -------------- A_N block ----------------
circuit.h(ctrl)  # H on control
ctrl_twos_complement(
    circuit, anc_qubits, target_qubits
)  # controlled two’s complement
\end{minted}
    \caption{Implementation of the unitary $A_N$.}
    \label{fig:unitary-A-impl}
\end{figure}

\section{Type-II Quantum Cosine and Sine Transforms}
\label{sec:qcst-II}

As with the Type-I transforms, Type-II transforms operate on quantum states of the form \(\ket{c} \otimes \ket{\psi}\), where \(\ket{c}\) is a single control qubit and \(\ket{\psi}\) is an \(n\)-qubit register. 

The algorithm for Type-II Quantum cosine and sine transforms \cite{klappenecker2001discrete} is based on the identity $U_N \cdot F_{2N} \cdot V_N = C_N^{\text{II}} \oplus (-i) S_N^{\text{II}}$, where \(F_{2N}\) is the Fourier transform. In the quantum setting, the above identity is expressed as
\begin{equation}
    \label{eq:qst-qct-II}
    U_N \cdot \qft_{2N} \cdot V_N = \ket{0}\bra{0} \otimes \qct_N^{\text{II}} - \ket{1}\bra{1} \otimes \qst_N^{\text{II}}.
\end{equation}
Therefore, given an $n$-qubit input state $\ket{\psi}$, similar to the Type-I transforms, we can start with $\ket{0}\ket{\psi}$ (resp. $\ket{1}\ket{\psi}$) and apply the unitary \eqref{eq:qst-qct-II} to obtain the state $\ket{0}\qct_N^{\text{II}}\ket{\psi}$ (resp. $\ket{1}\qst_N^{\text{II}}\ket{\psi}$). In the following, we briefly discuss the decomposition of the unitaries $U_N$ and $V_N$ into elementary gates appropriate for implementation.

The unitary \(V_N\), applied before the Fourier transform, prepares the control and data registers into the proper entangled form. It is decomposed as
\[
    V_N = P_{1C} (H \otimes \mathds{1}_N),
\]
where \(H\) is a Hadamard gate acting on the control qubit, and the conditional one's complement $P_{1C}$ uses the first qubit as control. The unitary \(U_N\), applied after the Fourier transform, disentangles the registers and aligns them with the output basis of the Type-II transform. It is decomposed as
\[
    U_N = \dec_n \cdot G \cdot P_{2C} \cdot D_1,
\]
where, both $P_{2C}$ and $\dec_n$ use the first qubit as control. The unitary \(G\) is a controlled entangling unitary acting across the control and data registers. The diagonal unitary \(D_1\) acts on the full register and plays a key role in producing the desired eigenstructure. It decomposes as $D_1 = (C \otimes \mathds{1}_N) \cdot (\Delta_1 \oplus \Delta_2)$, where the diagonal matrices $\Delta_1$ and $\Delta_2$ are defined as
\begin{align*}
    \Delta_1 & = \text{diag}(1, \omega_{4N}, \dots, \omega_{4N}^{N-1}), \\
    \Delta_2 & = \text{diag}(\omega_{4N}^{-N+1}, \dots, \omega_{4N}^{-1}, 1), \\
    C & = \text{diag}(1, \omega_{4N}^{-1}),
\end{align*}
These operators can be written as tensor products of elementary gates:
\begin{align*}
    \Delta_1 & = L_n \otimes \cdots \otimes L_1, \\
    \Delta_2 & = K_n \otimes \cdots \otimes K_1,
\end{align*}
where \(L_j = \text{diag}(1, \omega_{4N}^{2^{j-1}})\) and \(K_j = \text{diag}(\omega_{4N}^{-2^{j-1}}, 1)\). Finally, the unitary $G$ is defined by the following action:
\begin{align*}
    G \ket{00} &= \ket{00}, \\
    G \ket{0x} &= \frac{1}{\sqrt{2}} \ket{0x} + \frac{i}{\sqrt{2}} \ket{1x}, \\
    G \ket{10} &= -i \ket{10}, \\
    G \ket{1x} &= \frac{1}{\sqrt{2}} \ket{0x} - \frac{i}{\sqrt{2}} \ket{1x}.
\end{align*}
Algorithm \ref{alg:QCT_QST_Type_II} outlines the steps for the implementation of quantum cosine and sine transforms using the unitary in \eqref{eq:qst-qct-II}.

\begin{algorithm}[t]
    \caption{Type-II Quantum Cosine and Sine Transform}
    \label{alg:QCT_QST_Type_II}
    \begin{algorithmic}[1]
        \Require $(n + 1)$-qubit state $\ket{c}\ket{\psi}$, where $c \in \{0,1\}$.
        \Ensure $\ket{0}\qct_N^{\text{II}}\ket{\psi}$ if $c = 0$, and $\ket{1}\qst_N^{\text{II}}\ket{\psi}$ if $c = 1$.

        \State Apply $H\otimes \mathds{1}_N$
        \State Perform the conditional one's complement using the first qubit as control. \Comment{the unitary $P_{1C}$}
        \State Apply $\qft_{2N}$.
        \State Apply unitary $D_1$.
        \State Perform the conditional two's complement using the first qubit as control. \Comment{the unitary $P_{2C}$}
        \State Apply $H \otimes \mathds{1}_N$ and $S \otimes \mathds{1}_N$
        \State Apply the controlled unitaries $S^* \otimes \mathds{1}_N$, $H \otimes \mathds{1}_N$ and $S^* \otimes \mathds{1}_N$, conditioned on the last $n$-qubits being in the all-zero state $\ket{0^n}$.
        \State Using the first qubit as control, apply a controlled decrement-by-1. \Comment{the unitary $\dec_n$}
        \State Apply a $Z$ gate to the first qubit
    \end{algorithmic}
\end{algorithm}

\subsection{Implementation}

The quantum circuit implementation of \(D_1\) consists of applying the gates \(K_j\) and \(L_j\) to each data qubit, conditioned on the state of the control qubit. These gates are followed by the application of the single-qubit diagonal operator \(C\) to the control qubit, completing the implementation of the diagonal unitary shown in Figure \ref{fig:D_1-circuit}.
\begin{figure}
    \includegraphics[width = \columnwidth]{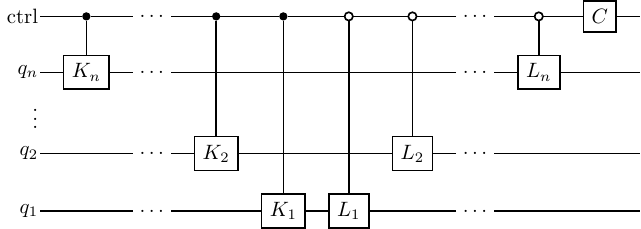}
    \caption{Circuit for the unitary $D_1$.}
    \label{fig:D_1-circuit}
\end{figure}
For the entangling operator \(G\), we leverage the same \orop-gate technique used in Section~\ref{sec:sine-cosine-I-impl} for the Type-I transforms to determine whether the data register contains a nonzero value. The result of this check is stored in an ancilla control qubit. We then apply a Hadamard gate followed by an \(S\) gate to the main control qubit. If the \orop-gate output evaluates to zero, we further apply the sequence of gates \(S^*\), \(H\), and another \(S^*\) to reverse the initial rotation. 

Notably, since the \orop-gate is used exactly once in both the computation and uncomputation of the ancilla qubits, it does not need to be duplicated. In our construction, the final ancilla qubit produced by the \orop-gate is directly used to store the result, so no further \orop-gate evaluation is required to reset it. Therefore, a total of \(12(n - 1)\) gates is sufficient to perform both the computation and ancilla recovery, where \(n\) is the number of data qubits. Figure \ref{fig:unitary-G-impl} shows the implementation of the unitary $G$ as a gate.

\begin{figure}
\begin{minted}{python}
def G_gate(circuit: QuantumCircuit, 
           target_qubits: list[int], 
           anc_qubits: list[int]):
    # control qubit
    ctrl = target_qubits[-1]
    circuit.h(ctrl)
    circuit.s(ctrl)  

    # apply when x is zero
    circuit.x(anc_qubits[0])  
    # apply control S-dagger gate
    circuit.csdg(anc_qubits[0], ctrl)  
    circuit.ch(anc_qubits[0], ctrl)  
    circuit.csdg(anc_qubits[0], ctrl)
    circuit.x(anc_qubits[0])
\end{minted}
    \caption{Implementation of the unitary $G$. The ancilla qubit contains the result of the \orop-gate evaluation.}
    \label{fig:unitary-G-impl}
\end{figure}

\section{Type-IV Quantum Cosine and Sine Transforms}

Similar to other cosine and sine transforms, Type-IV transforms operate on quantum states of the form \(\ket{c}\ket{\psi}\), where \(\ket{c}\) is a single control qubit and \(\ket{\psi}\) is an \(n\)-qubit register. The algorithm for Type-IV transforms \cite{klappenecker2001discrete} is based on the identity
\[
    M \cdot U_N^T \cdot F_{2N} \cdot U_N = C_N^{\text{IV}} \oplus (-i) S_N^{\text{IV}},
\]
where \(F_{2N}\) is the Fourier transform. In the quantum setting, the above identity is expressed as
\begin{equation}
    \label{eq:qst-qct-IV}
    M \cdot U_N^T \cdot \qft_{2N} \cdot U_N = \ket{0}\bra{0} \otimes \qct_N^{\text{IV}} - i\ket{1}\bra{1} \otimes \qst_N^{\text{IV}},
\end{equation}
Similar to the other transforms, setting the ancilla qubit $\ket{c}$ in the input $\ket{c}\ket{\psi}$ determines whether the $\qct_N^{\text{IV}}$ ($c = 0$) or $\qst_N^{\text{IV}}$ ($c = 1$) is applied to the $n$-qubit state $\ket{\psi}$. In the following, we briefly explain the unitaries $M$ and $U_N$.  

The diagonal unitary \(M = \mathrm{diag}(\omega_{4N}, \omega_{4N}) \otimes \mathds{1}_N\) is a global phase gate that acts only on the control qubit, ensuring that the overall output phase matches the canonical form of the Type-IV basis. The unitary \(U_N\) is structured to entangle the control and data registers in a carefully aligned phase basis. It begins with two single-qubit gates applied to the control: a Hadamard gate followed by an \(S^*\) gate. This creates a simple superposition with an embedded global phase. We can decompose $U_N$ as 
\[ U_N = P_{1C} D_2 ((HS^*) \otimes \mathds{1}_N), \]
where $P_{1C}$ uses the first qubit as control.

The unitary $D_2$ is a diagonal unitary that jointly acts on the control and data registers and is responsible for introducing data-dependent phases conditional on the control qubit. It can be decomposed as $D_2 = (C\otimes \mathds{1}_N) (\Delta_1 \oplus \Delta_1^*)$, where $\Delta_1 = \text{diag}(1, \omega_{4N}, \dots, \omega_{4N}^{N-1})$ and $C = \text{diag}(1, \omega_{4N}^{-1})$ are the unitaries also used in type-II transforms\footnote{In the original algorithm of \cite{klappenecker2001discrete}, the unitary $D_2$ is written as $D_2 = (C\otimes \mathds{1}_N) (\Delta_1 \oplus \Delta_2)$. This is not correct, i.e., it does not lead to the identity in \eqref{eq:qst-qct-IV}. We show in Appendix \ref{sec:correct-IV} that $D_2 = (C\otimes \mathds{1}_N) (\Delta_1 \oplus \Delta_1^*)$.}. Algorithm \ref{alg:QCT_QST_Type_IV} outlines the steps for implementing quantum cosine and sine transforms using the unitary in \eqref{eq:qst-qct-IV}.

\begin{algorithm}
    \caption{Type-IV Quantum Cosine and Sine Transform}
    \label{alg:QCT_QST_Type_IV}
    \begin{algorithmic}[1]
        \Require $(n + 1)$-qubit state $\ket{c}\ket{\phi}$ with $c\in\{0,1\}$.
        \Ensure $\ket{0} \qct_{N}\ket{\phi}$ if $c=0$, and $\ket{1}\qst_{N}\ket{\phi}$ if $c = 1$
        \State Apply $S^*\otimes \mathds{1}_N$ and $H\otimes \mathds{1}_N$
        \State Apply the unitary $D_2$ gate. 
        \State Perform the conditional one's complement using the first qubit as control. \Comment{the unitary $P_{1C}$}
        \State Apply $\qft_{2N}$.
        \State Perform the conditional one's complement using the first qubit as control. \Comment{the unitary $P_{1C}$}
        \State Apply the unitary $D_2^*$.
        \State Apply $H\otimes \mathds{1}_N$ and $S^* \otimes \mathds{1}_N$
        \State Apply $M$, where $M = \mathrm{diag}(\omega_{4N}, \omega_{4N}) \otimes \mathds{1}_N$.
        \State Apply an $S$ gate to the first qubit.        
    \end{algorithmic}
\end{algorithm}

\subsection{Implementation}

The unitary $D_2$ is constructed using the gates $L_j$ and $L_j^*$, $j = 1, \dots, n$, defined in Section \ref{sec:qcst-II}. The circuit for $D_2$, as shown in Figure \ref{fig:D_2-circuit}, closely resembles that of the unitary \(D_1\) defined in the Type-II transform, but with the gates $K_j$ replaced by $L_j^*$ for $j = 1, \dots, n$.

\begin{figure}[H]
    \includegraphics[width = \columnwidth]{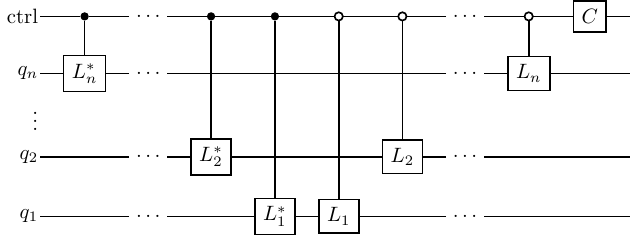}
    \caption{Circuit for the unitary $D_2$.}
    \label{fig:D_2-circuit}
\end{figure}

\noindent The snippet in Figure \ref{fig:Li-Ki-impl} shows the implementation of the gates $L_i$ and $K_i$.

\begin{figure}[H]
\begin{minted}{python}
def K_i(circuit: QuantumCircuit,
        control_qubit: int,
        target_qubit: int,
        i: int,
        theta: float):

    circuit.x(target_qubit)
    L_i(circuit, control_qubit, target_qubit, i, 
        -theta)
    circuit.x(target_qubit)

def L_i(circuit: QuantumCircuit,
        control_qubit: int,
        target_qubit: int,
        i: int,
        theta: float):

    circuit.cp((2 ** (i - 1)) * theta, 
               control_qubit, target_qubit)
\end{minted}
    \caption{Implementation of the diagonal gates $K_i$ and $L_i$ used in the unitary $D_1$ and $D_2$.}
    \label{fig:Li-Ki-impl}
\end{figure}

\appendix

\section{Linear Combination of Unitaries (LCU)}
\label{sec:LCU}

In this section, we briefly review the LCU technique, following the expositions in \cite{childs2012hamiltonian, kothari2014efficient}. 
Let $U_1, \dots, U_M$ be a set of unitaries acting on an $N$-dimensional Hilbert space $\X_1$, corresponding to an $n$-qubit system. 
For simplicity, assume $M = 2^m$ for some integer $m > 0$. 

Given an operator $V = \sum_{j=0}^{M-1} a_j U_j$, where $a_j \in \R$, the idea of LCU is to design a simple quantum circuit that can implement the action of $V$. 
If $V$ is not unitary, any implementation of $V$ will necessarily be probabilistic, i.e., for any state $\ket{\psi} \in \X_1$, the state $V\ket{\psi}$ is obtained only with a certain probability. 
If $V$ is unitary, however, one can achieve a deterministic (i.e., probability-1) implementation of $V$ using a procedure called \emph{oblivious amplitude amplification}. 
Throughout this section, we assume $V$ is unitary.

Let \( a := \sum_j a_j \). Define the unitary \( A \) on an $M$-dimensional Hilbert space $\X_2$ by
\begin{equation}
    \label{eq:luc-A}
    A \ket{0^m} = \frac{1}{\sqrt{a}} \sum_{j=0}^{M-1} \sqrt{a_j}\, \ket{j}.
\end{equation}
Let $U := \sum_{j=0}^{M-1} \ket{j}\bra{j} \otimes U_j$ be the block-diagonal unitary encoding the $U_j$. Define 
\[
    W := (A^* \otimes \mathds{1}) U (A \otimes \mathds{1})
    \quad \text{and} \quad 
    \Pi := \ket{0^m}\bra{0^m} \otimes \mathds{1},
\]
both acting on the space $\X_2 \otimes \X_1$. The following lemma gives a probabilistic implementation of $V$.

\begin{lemma}[Lemma 2.1 of \cite{kothari2014efficient}]
    \label{lem:lcu-enc}
    For all $n$-qubit states \( \ket{\psi} \in \X_1 \), we have
    \[
        W \ket{0^m} \ket{\psi} = \frac{1}{a} \ket{0^m} V \ket{\psi} + \ket{\Phi^\perp},
    \]
    where the state \( \ket{\Phi^\perp} \in \X_2 \otimes \X_1 \) depends on $\ket{\psi}$ and satisfies $\Pi \ket{\Phi^\perp} = 0$.
\end{lemma}

According to Lemma \ref{lem:lcu-enc}, to compute $V\ket{\psi}$, we first apply $W$ and then measure the first register. 
If the outcome is $\ket{0^m}$, then the resulting state is $V\ket{\psi}$. 
Since $V$ is unitary, the probability of success for this procedure is $1/a^2$. 
To boost this probability, we can perform a version of amplitude amplification, stated in the following lemma \cite[Lemma 2.2]{kothari2014efficient}.

\begin{lemma}[Oblivious amplitude amplification]
    \label{lem:obl-amp}
    Let $V$ be a unitary on an $n$-qubit space $\X_1$ and let $\theta \in (0, \pi / 2)$. 
    Let $W$ be a unitary on the $(m + n)$-qubit space $\X_2 \otimes \X_1$ such that for all $\ket{\psi} \in \X_1$,
    \[
        W \ket{0^m} \ket{\psi} = \sin(\theta) \ket{0^m} V \ket{\psi} + \cos(\theta) \ket{\Phi^\perp},
    \]
    where the $(m + n)$-qubit state \( \ket{\Phi^\perp} \in \X_2 \otimes \X_1\) depends on $\ket{\psi}$ and satisfies $\Pi \ket{\Phi^\perp} = 0$. 
    Let $R := 2\Pi - \mathds{1}$ and define $S := -WRW^*R$. Then for any $k \in \Z$,
    
    \begin{align*}
        S^k W \ket{0^m} \ket{\psi}
        &= \sin((2k + 1)\theta) \ket{0^m} V \ket{\psi} \\
        &\mathrel{\phantom{=}}+ \cos((2k + 1)\theta) \ket{\Phi^\perp}.
    \end{align*}
\end{lemma}

Combining Lemmas \ref{lem:lcu-enc} and \ref{lem:obl-amp} gives a procedure for implementing any $V$ that is a linear combination of unitaries. 
From the action of $W$, we first find $\theta$ such that $\sin(\theta) = 1/a$. 
If $(2k + 1)\theta = \pi/2$ for some integer $k$, then using this $k$ we obtain an exact implementation of $V$:
\[
    S^k W \ket{0^m} \ket{\psi} = \ket{0^m} V \ket{\psi},
\]
a unitary operation requiring an $m$-qubit ancilla. 

If $\pi / (2\theta)$ is not an odd integer, let $2k + 1$ be the smallest odd integer larger than $\pi / (2\theta)$. 
Then there exists $\theta' < \theta$ such that $(2k + 1)\theta' = \pi/2$. 
Define the rotation
\[
    P: \ket{0} \mapsto \Big( \frac{\sin\theta'}{\sin\theta} \Big) \ket{0} + \sqrt{1 - \Big( \frac{\sin\theta'}{\sin\theta} \Big)^2} \ket{1}.
\]
Now construct the new unitary $W' = P \otimes W$, which acts as
\[
    W'\ket{0^{m+1}} \ket{\psi} = \sin(\theta')\ket{0^{m+1}} V \ket{\psi} + \cos(\theta')\ket{\Phi^\perp},
\]
and replace the original $W, R, S$ in the amplitude amplification procedure with $W', S', R'$, where $R' = 2\Pi' - \mathds{1}$ and $\Pi' = \ket{0^{m + 1}}\bra{0^{m + 1}}$. Therefore, using the new $k$ and $\theta'$, and an extra ancilla qubit, we achieve an exact implementation of $V$:
\[
    (S')^k W' \ket{0^{m+1}} \ket{\psi} = \ket{0^{m+1}} V \ket{\psi}.
\]

\section{Correction on the Type-IV Transforms}
\label{sec:correct-IV}

In the original algorithm for type-IV transforms proposed in \cite{klappenecker2001discrete}, the unitary
\[ D_2 = (C\otimes \mathds{1}_N) (\Delta_1 \oplus \Delta_2) \]
is used, where
\[ \Delta_2 = \text{diag}\,(\omega_{4N}^{-N + 1}, \dots, \omega_{4N}^{-2}, \omega_{4N}^{-1}, 1). \]
In this section, we briefly show, by direct calculation, that the correct unitary for $D_2$ is given by $ D_2 = (C\otimes \mathds{1}_N) (\Delta_1 \oplus \Delta_1^*)$ instead.

Recall that the unitary $U_N$ admits the decomposition
\begin{align*}
    U_N
    & = P_{1C} D_2 ((HS^*) \otimes \mathds{1}_N) \\
    & = P_{1C} (C\otimes \mathds{1}_N) (\Delta_1\oplus\Delta_1^*) ((HS^*) \otimes \mathds{1}_N),
\end{align*}
In matrix notation, we have
\[
    (HS^*) \otimes \mathds{1}_N = \frac{1}{\sqrt{2}}
    \begin{bmatrix}
        \mathds{1}_N & -i\mathds{1}_N \\
        \mathds{1}_N & i\mathds{1}_N
    \end{bmatrix}.
\]
We also have
\begin{align*}
    \Delta_1\oplus\Delta_1^* & =
    \begin{bmatrix}
        \Delta_1 & 0 \\
        0 & \Delta_1^*
    \end{bmatrix}, \\[2mm]
    C\otimes \mathds{1}_N & =
    \begin{bmatrix}
        \mathds{1}_N & 0 \\
        0 & \omega_{4N}^{-1} \mathds{1}_N
    \end{bmatrix}, \\[2mm]
    P_{1C} & =
    \begin{bmatrix}
        \mathds{1}_N & 0 \\
        0 & X^{\otimes n}
    \end{bmatrix}.
\end{align*}
Multiplying the factors step by step gives
\begin{align*} 
    U_N
    & = P_{1C} (C\otimes \mathds{1}_N) (\Delta_1\oplus\Delta_1^*) ((HS^*) \otimes \mathds{1}_N) \\
    & = P_{1C} (C\otimes \mathds{1}_N) (\Delta_1\oplus\Delta_1^*)
    \begin{bmatrix}
        \mathds{1}_N & -i\mathds{1}_N \\
        \mathds{1}_N & i\mathds{1}_N
    \end{bmatrix} \\    
    & = P_{1C} (C\otimes \mathds{1}_N) 
    \begin{bmatrix}
        \Delta_1 & -i\Delta_1 \\
        \Delta_1^* & i\Delta_1^*
    \end{bmatrix} \\
    & = P_{1C} \frac{1}{\sqrt{2}} 
    \begin{bmatrix}
        \Delta_1 & -i\Delta_1 \\
        \omega_{4N}^{-1} \Delta_1^* & i\omega_{4N}^{-1} \Delta_1^*
    \end{bmatrix} \\
    & = \frac{1}{\sqrt{2}}
    \begin{bmatrix}
        \Delta_1 & -i\Delta_1 \\
        \omega_{4N}^{-1} X^{\otimes n} \Delta_1^* & i\omega_{4N}^{-1} X^{\otimes n} \Delta_1^*
    \end{bmatrix}.
\end{align*}
A direct comparison shows that the explicit form of the matrix $U_N$ is
\[
    \setlength{\arraycolsep}{0.5mm}
    U_N = \frac{1}{\sqrt{2}}
    \begin{bmatrix}
        1 &             &             & -i              &             &             &             & \\
        & \omega_{4N}      &             &                 & -i\omega_{4N}    &             &             & \\
        &             & \ddots      &                 &             & \ddots      &             & \\
        &             &             & \omega_{4N}^{N-1}    &             &             & -i\omega_{4N}^{N-1} & \\
        &             &             & \omega_{4N}^{-N} &      &             & 1           & \\
        &             & \iddots     &                 &             & \iddots     &             & \\
        & \omega_{4N}^{-2}   &                 &             & i\omega_{4N}^{-2} & & & \\
        \omega_{4N}^{-1} &           &                 & i\omega_{4N}^{-1}             &  &      & &
    \end{bmatrix},
\]
which agrees with identity \eqref{eq:qst-qct-IV}.

\bibliographystyle{plain}
\bibliography{references}

\end{document}